\documentclass[a4paper,twocolumn,10pt,reqno,nonamelimits]{article}
\usepackage{newcent}         
\usepackage{nimbusmono}

\usepackage{amsmath}
\usepackage{amssymb}
\usepackage{amsfonts}
\usepackage{amscd}
\usepackage{amsthm}
\usepackage{amsxtra}
\usepackage{stmaryrd}
\usepackage{mathrsfs}
\usepackage{nicefrac}
\usepackage{graphicx}
\usepackage{xcolor}
\usepackage{tikz}
\usepackage{enumitem}
\setlist[enumerate]{leftmargin=2em,itemindent=0em, labelindent=0pt,labelwidth=1.5em,labelsep=.5em, align=left, noitemsep}
\newlist{txtenum}{enumerate}{1}
\setlist[txtenum]{leftmargin=0em,itemindent=1.5em, labelindent=0pt,labelwidth=1em,labelsep=.5em, align=left}

\usepackage[norelsize,boxed,noend,linesnumbered]{algorithm2e}\DontPrintSemicolon
\SetKwFor{RepeatTimes}{repeat}{times}{endrepeat}
\RestyleAlgo{ruled}
\SetKwInOut{Input}{Input}\SetKwInOut{Output}{Returns}

\usepackage{qcircuit}

\usepackage{datetime}

\theoremstyle{plain}
\newtheorem{theorem}{Theorem}
\newtheorem*{theorem*}{Theorem}

\newtheorem*{proposition*}{Proposition}

\newtheorem*{corollary*}{Corollary}

\newtheorem{lemma}[theorem]{Lemma}
\newtheorem*{lemma*}{Lemma}

\newtheorem*{observation*}{Observation}

\newtheorem*{conjecture*}{Conjecture}

\newtheorem*{question*}{Question}

\newtheorem*{questions*}{Questions}

\newtheorem*{problem*}{Problem}

\newtheorem*{problems*}{Problems}

\newtheorem*{openproblem*}{Open Problem}

\theoremstyle{definition}
\newtheorem{definition}[theorem]{Definition}
\newtheorem*{definition*}{Definition}

\newtheorem*{example*}{Example}

\newtheorem*{exercise*}{Exercise}

\newtheorem*{remark*}{Remark}

\newtheorem*{remarks*}{Remarks}

\theoremstyle{remark}

\newtheorem*{claim*}{Claim}

\newcommand{\subclass}[1]{}

\newcommand{\enumTi}[1]{\renewcommand{\theenumi}{#1}}

\newcommand{\alphenumi}{\enumTi{\alph{enumi}}}

\newcommand{\romenumi}{\enumTi{\roman{enumi}}}

\newlength{\hspaceforlengthglumpf}

\renewcommand{\em}{\sl}

\DeclareMathOperator{\tr}{tr}

\newcommand{\lt}{\left}
\newcommand{\rt}{\right}

\newcommand{\abs}[1]{{\lt\lvert{#1}\rt\rvert}}

\newcommand{\nfrac}[2]{{\nicefrac{#1}{#2}}}

\newcommand{\NN}{\mathbb{N}}

\newcommand{\RR}{\mathbb{R}}

\newcommand{\TT}{\mathbb{T}}
\newcommand{\ZZ}{\mathbb{Z}}

\DeclareMathOperator*{\Prb}{\mathbb{P}}
\DeclareMathOperator*{\Exp}{\mathbb{E}}

\newcommand{\bra}[1]{{\lt< #1 \rt|}}
\newcommand{\ket}[1]{{\lt| #1 \rt>}}

\newlength{\algotabbingwidth}
\renewcommand{\paragraph}[1]{\medskip\noindent{\textsl{#1.}}}

\newcommand{\sq}{\boxempty}

\newcommand{\Algo}{\textup{\textbf{\texttt{A}}}}

\newcommand{\Ee}{\mathcal E}

\newcommand{\Tn}[1][n]{\TT^{#1}}
\newcommand{\Tone}{\TT}

\newcommand{\dTn}[1][n]{\{0,\nfrac13,\nfrac23\}^{#1}}

\newcommand{\mutinf}{\mathbb I}
\newcommand{\entropy}{\mathbb H}

\newcommand{\law}[1]{\mathcal{#1}}
\newcommand{\lF}{\law{F}}

\begin{document}
\title{Information content of queries in training Parameterized Quantum Circuits}
\author{Evgenii Dolzhkov$^{a}$ \and Bahman Ghandchi$^{a,b}$ \and Dirk Oliver Theis$^{a,b}$
 \\[1ex]
   \small $^a$ Theoretical Computer Science, University of Tartu, Estonia\\
   \small $^b$ Ketita Labs {\tiny O\"U}, Tartu, Estonia\\
   \small \texttt{ghandchi@}\{\texttt{ketita.com}, \texttt{ut.ee}\}, \texttt{dotheis@}\{\texttt{ketita.com}, \texttt{ut.ee}\}%
}
\date{Version: Fri Mar 29 17:49:52 CET 2019
  \\
  Compiled: \currenttime}
\maketitle

\begin{abstract}\sl%
  Parameterized quantum circuits (PQC, aka, variational quantum circuits) are among the proposals for a computational advantage over classical computation of near-term (not error corrected) digital quantum computers.  PQCs have to be ``trained'' --- i.e., the expectation value function has to be maximized over the space of parameters.

  This paper deals with the number of samples (or ``runs'' of the quantum computer) which are required to train the PQC, and approaches it from an information theoretic viewpoint.  The main take-away is a disparity in the large amount of information contained in a single exact evaluation of the expectation value, vs the exponentially small amount contained in the random sample obtained from a single run of the quantum circuit.
  \par\small%
\end{abstract}

\section{Introduction}\label{sec:intro}
Hybrid quantum-classical computing with parameterized (or variational) quantum circuits (PQCs) works by alternately running the parameterized quantum circuit on a digital, gate-based quantum computer, and updating parameters in classical hard- and software.  The hybrid process aims to find parameter settings which optimize some objective function derived from the measurement results of the PQC, for example with the goal to find the ground state of the measurement Hamiltonian.  This process has become known as ``training'' the PQC.

For the purpose of this paper, a Parameterized Quantum Circuit consists of a sequence of quantum operations, applied to a known initial state which we denote by $\ket0\bra0$, and followed by a measurement.  Some of the quantum operations are unitaries of the form
\begin{equation}\label{eq:ham-unitary}
  \rho \mapsto e^{-\pi i x_j H_j}\rho e^{\pi i x_j H_j}, \text{ $j=1,\dots,n$,}
\end{equation}
where the $H_j$ are hermitian operators and $x\in\RR^n$ is the vector of parameters.  For simplicity(!), we assume that the $H_j$ have $\pm1$ eigenvalues.  (We allow more general dependence on the parameters in Section~\ref{sec:eval-queries}.)  We also assume that the observable in the final measurement has eigenvalues $\pm1$.  Hence, a single run of the quantum circuit (with measurement) with parameters set to~$x$ yields a random number in $\pm1$, whose expectation we denote by $f(x)$, and refer to as the \textit{expectation value function} of the PQC.  In this simplified setting, the training problem is this:
\begin{equation*}
  \text{maximize } f(x) \text{ over $x\in\RR^n$.}
\end{equation*}
(Note that $n$ is the number of parameters, not the number of qubits.)

Even though, in applications, a good local maximum is often sufficient, training PQCs is known to be difficult for a variety of reasons.  The least of it is that, as a non-concave maximization\footnote{Sorry.  Next time, we'll do ``minimize'' and ``convex''.} problem, the training problem is likely to be NP-hard:  But classical neural network training has the same property, and it is not a huge problem there.  More specific to the quantum case is the existence of ``plateaus'': large regions of the parameter space where the gradient is close to~0 \cite{McClean-Boixo-Smelyanskiy-Babbush-Neven:barren:2018}.  While training seems to work fine in practice with a small number of qubits, the exponential dependence on the number of qubits of indicators of ``trouble'' are worrysome.
In this paper, we add one new worrying perspective to the discussion: The information content of the random output of a run of the PQC.  For that we consider a setting which is very generous to the designer of a training algorithm: The algorithm is only ever used on a fixed~$n$, and a fixed, finite number (depending on~$n$) of functions $f_c$, $c\in\mathcal C$, all of which are known to the algorithm.  The algorithm itself can be randomized.  The algorithm has infinite computational resources; e.g., it can represent real numbers exactly, and make instantaneous computations on them (for parameters and expectation values).

Formally, we define the following.  A \textit{sample query} consists of the training algorithm setting the parameters to an $x\in\RR^n$ \textsl{ad libitum}, and then running once the quantum circuit with this setting, retrieving the resulting random number $F \in \{\pm1\}$ with $\Exp F=f(x)$.  In contrast, in an \textit{evaluation query} after setting the parameters \textsl{ad libitum,} the algorithm is given the real number $f(x)$ exactly.

The success of the algorithm is only measured in some definition of average-case --- not worst-case --- over $c\in \mathcal C$ and over its internal randomness.
The following algorithm and theorem underlines just how ridiculously generous the compuational model is.

\begin{algorithm}
  \TitleOfAlgo{Omnipotent-Algorithm}
  \Input{Evaluation-query access to a function $f\in\{f_c\mid c\in\mathcal C\}$}
  \Output{$x\in\RR^n$ maximizing $f(x)$}

  Pick a random parameter setting $x\in\RR^n$ \label{step:random-x}\\
  Query $f(x)$. \\
  Iterating over all $c\in \mathcal C$, find one with $f_c(x) = f(x)$\\
  Look up the parameter setting $x^*$ maximizing $f_c$ in a table\\
  output $x^*$.
\end{algorithm}

\begin{theorem}\label{thm:super}
  Let $\mu$ be an arbitrary absolutely continuous probability measure on $\RR^n$.
  If in Step~\ref{step:random-x}, $x$ is drawn according to~$\mu$, then the Omnipotent Algorithm succeeds with probability~1.
\end{theorem}

The proof of this theorem, in Section~\ref{sec:eval-queries}, will show that with probability~1 over the choice of~$x$, the mapping $f_c \mapsto (x,f_c(x))$ is one-to-one.

In infomation theoretic terms, if $C$ is randomly chosen in $\mathcal C$, then a single evaluation query of $f_C$ at a random point contains all information about~$f_C$:
\begin{equation}\label{eq:eval-q:cond-ent-0}
  \entropy\bigl(f_A \mid (X,f_A(X)) \bigr) =0.
\end{equation}

Our \textsl{Superman} algorithms with infinite memory and tables with worked-out solutions hit Kryptonite, when we replace evaluation-query access by sample-query access --- even when we allow the output to be off by a significant amount from the true maximum.  We propose the the following definition.

\begin{definition}
  Let \Algo{} be a Las-Vegas (randomized) algorithm which is given sample-query access to one (unknown) element in a family~$\mathcal C$ of PQCs with~$n$ parameters.  For $\alpha\in\RR_+$, we say that \Algo{} \textit{$\alpha$-succeeds} in the training problem, if it outputs a parameter setting $x^*$ with $f(x^*) \ge \max_x f(x) -\alpha$, after performing a number of sample quries whose number depends on: the internal randomness of \Algo{}; the random choice of $c$ uniformly in $\mathcal C$; and the randomness in the sample query results.
\end{definition}

\begin{theorem}\label{thm:sample}
  There are constants $c>1$ and $\alpha \approx 1$ and a family of simple PQCs ($3^n$ for each number of qubits~$n$) such that every $\alpha$-successful training algorithm, with probability $1-c^{-n}$, requires at least $c^n$ sample queries.
\end{theorem}

In terms of information content of queries: For any $m \ll c^n$, for the queries at the (random) parameter settings $X_1,\dots,X_m$ that a randomized training algorithm performs, and the corresponding sample query results $Q_1,\dots,Q_m$, if~$C \in \mathcal C$ is chosen uniformly at random, we have
\begin{equation}\label{eq:sample-q:mutinf-exp-small}
  \mutinf\bigl(C : (X_1,Q_1,\dots,X_m,Q_m) \bigr) \le m 2^{-\Omega(n)}\entropy(C).
\end{equation}

Clearly, the difference in the perfect performance of the algorithms for evaluation and sample queries lies in the information content of the queries, \eqref{eq:eval-q:cond-ent-0} vs \eqref{eq:sample-q:mutinf-exp-small}.  The contribution of this paper lies in bringing the consideration of information content of queries to the table with regards to algorithms and lower bounds for PQC training.

\paragraph{This paper is organized as follows}
In Section~\ref{sec:eval-queries}, we prove Theorem~\ref{thm:super}, Section~\ref{sec:sample-queries} is dedicted to the proof of Theorem~\ref{thm:sample}.  We conclude with an outlook on related questions.  We shoved some technicalities out of the weary eye of the reader into the appendix.

\section{Evaluation queries are powerful}\label{sec:eval-queries}
We consider a family of parameterized quantum circuits which is more general than in the rest of the paper: The only restriction is that the parameters $x\in\RR^m$ occur only in unitary quantum operations of the form
\begin{equation*}
  \ket{\psi} \mapsto e^{-i \sum_j \theta_j(x) H_j} \ket{\psi},
\end{equation*}
for arbitrary Hermitian operators $H_j$ and~$H$ on no matter how many qubits, and arbitrary analytic functions $\theta_j\colon\RR^m \to \RR$.

In this section, we prove the following theorem
\begin{theorem}
  Let $\mathscr C$ be a countable family of parameterized quantum circuits, each taking exactly $m$ parameters, and such that the expectation value functions of the circuits are all distinct.

  If $x\in\RR^m$ is chosen randomly according to an absolutely continuous measure on $\RR^m$, then, with probability~1, the circuit $C\in\mathscr C$ is uniquely determined by a single query of its expectation value (evaluation query) with parameters~$x$.
\end{theorem}

\begin{proof}
  The expectation value function is real analytic.  Real analytic function have the property that zero-sets are lower-dimensional sub-varieties.  Hence, their measure, wrt any absolutely continuous probability measure, is $0$.  Consequently, the set of points~$x$ for which two such functions coincide also has measure zero.  We refer to \cite{Dolzhkov:MSc:2020} for the details.
\end{proof}

\section{Sample-queries are poor}\label{sec:sample-queries}
\paragraph{We will need some notation} %
It can be seen~\cite{GilVidal-Theis:CalcPQC:2018} that $f$ is periodic in each coordinate, so we set $\Tone := \RR/\ZZ$; note that $\Tn = (\RR/\ZZ)^n = \RR^n/\ZZ^n$.  



For a function~$f$ with range in $[-1,+1]$, we write $\lF_f(x)$ for the probability distribution with support contained in $\{\pm1\}$ and mean $f(x)$.  Modeling the behavior of parameterized quantum circuits, we assume that, conditioned on $f$ and~$x$, samples from $\lF$ are independent from each other and from all other random variables.

For any function $f_0\colon\Tn\to\RR$ and $a \in \dTn \subset \Tn$, we define the function $f_a := f_0(\sq -a)$; similarly, for a set $P_0 \subset \Tn$, we define $P_a := \{ x-a \mid x \in P_0\}$.

We are now well equipped to describe the technical approach.
\subsection{The ``plateau game''}\label{ssec:sample-queries:game}
For $n\in \NN$, consider the following \textit{plateau game}\footnote{This isn't really a game, it's an active learning problem --- but games are so much more fun\dots}, played between Alice and Bob:
\begin{equation}\label{TheGame}\tag{$\mathscr G_n$}
    \parbox{.85\linewidth}{\sl%
      There is a set $P_0 \subset \Tn$ with $0\notin P_0$, which is known to both Alice and Bob.

      At the beginning of the game, Alice chooses an\footnotemark{} $a\in \dTn$, hidden from Bob.
      Then, the players proceed in rounds.  In each round, Bob chooses an element $x\in\Tn$, and asks Alice, ``Is $x \in P_{a}$?'', to which Alice answers truthfully.

      The game ends when Bob queries a point which is not in $P_{a}$, in which case Bob wins; otherwise the game proceeds forever (and Alice wins).
    }%
\end{equation}
\footnotetext{At this point, it deserves to be emphasizes that all results hold for other integers $k \ge 3$ (with different constants in the estimates); the use of ``$3$'' is a convenience.}
Bob can make randomized queries, i.e., Alice cannot ``change her mind'' about her choice of~$a$ after the first round begins.

Suppose that Bob makes randomized queries and wins after~$M_a$ rounds, in the case that Alice picks~$a$ as her hidden point ($M_a$ is a random variable).
Suppose further that Alice chooses her hidden point uniformly at random in $\dTn$; since it is a random variable, we denote it by a capital letter, $A$.
Finally, let
\begin{equation}\label{eq:def-p}
  p := \sup_{x\in\Tn} \Prb\bigl( P_A \not\ni x \bigr)
\end{equation}
denote the highest probability of Bob winning in the first round.

Here is the lower bound on the number of rounds it takes Bob to succeed.
\begin{lemma}\label{lem:stupid-game}
  $\displaystyle
    \Prb\bigl( M_A \le m ) \le pm
    $
\end{lemma}
Before we give the proof, a remark.  While the setting of the game allows Bob to adaptively choose his next move based on Alice's answers to his previous questions, that adaptivity is not really present, as Bob --- unless he has already won --- (knows that he) always gets the same answer: ``Yes''.  Indeed, the transition from the training problem to the game (performed in Lemma~\ref{lem:tech} below) takes the adaptivity out of the equation.

Hence, Bob's strategy is a single probability measure, which he may just as well choose before Alice picks her point.  We refer to Appendix~\ref{apx:non-adapt} for the formal justification.

\begin{proof}[Proof of Lemma~\ref{lem:stupid-game}]
  Let $m \in\NN$.  To decide $M_A \le m$, we need to observe the game only for the first (at most) $m$ rounds.  For $x\in (\Tn)^m$, $a\in\dTn$, we let $\beta(a,x) := 1$, if Bob wins in the first~$m$ rounds with the questions $x_1,x_2,\dots,x_m$, i.e., if there is a $j$ with $x_j\not\in P_a$; and $\beta(a,x) := 0$ otherwise.

  By the above remark, the set of Bob's strategies is the set of probability distributions~$\mu$ on $(\Tn)^m$, and $\mu$ does not depend on~$a$.  Now we just compute (the sums are over all $\in\dTn$):
  \begin{align*}
    \Prb\bigl( M_A \le m )
    &=
      \tfrac{1}{3^n} \sum_a  \int_{(\Tn)^m} \beta(a,x) \,d\mu(x)
    \\
    &=
      \int_{(\Tn)^m} \tfrac{1}{3^n}\sum_a \beta(a,x) \,d\mu(x)
    \\
    &=
      \int_{(\Tn)^m} \Prb( \exists j\colon x_j\notin P_a ) \,d\mu(x)
    \\
    &\le
      m\cdot \sup_{x\in (\Tn)^m} \Prb(x_j\notin P_a )
    \\
    &=
      m p.
  \end{align*}
  This completes the proof of Lemma~\ref{lem:stupid-game}.
\end{proof}

\subsection{From the training problem to the plateau game}\label{ssec:sample-queries:core-lemma}
Now we come to the technical lemma which connects the PQC-training problem to the plateau game~\eqref{TheGame}.
\begin{lemma}\label{lem:tech}
  Let $n \in \NN$,
  $f_0\colon \Tn \to \lt[ -1 , +1 \rt]$ a continuous function, and $P_0\subset\Tn$.
  Further, let $\eta \in [-1,+1]$, $0\le \delta < \alpha$ (we will need $\delta \ll 1$ and $\alpha \gg \delta$ for the conclusion to make sense).
  We assume
  \begin{enumerate}[label=(\roman*)]
  \item\label{lem:tech:cond:near-const} $\displaystyle |f_0(x) - \eta| < \delta$ for all $x\in P_0$
  \item\label{lem:tech:cond:max} $\displaystyle f_0(0) = \max_x f_0(x)$;
  \item\label{lem:tech:cond:diff} $\max_{x\in P_0} f_0(x) + \alpha < f_0(0)$ for an $\alpha \gg \delta$.
  \end{enumerate}
  Consider an $\alpha$-successful algorithm \Algo{} for maximizing a function $f\in \{f_a \mid a\in \dTn\}$, which has sample-query access to~$f$.  With $A$ chosen uniformly at random from $\dTn$, denote the number of sample queries performed by the algorithm for maximizing $f=f_A$ before it terminates by the random variable $T_A$.

  With~$p$ as in~\eqref{eq:def-p}, we have that
  \begin{equation}\label{eq:ub-on-nqueries}
    \Prb( T_A \le m ) \le (p+\delta/2) m  +1
  \end{equation}
\end{lemma}

The proof fills the remainder of this subsection.
The idea is that we want to compare what the algorithm does when it samples according to $\lF_f$ to what it does when it samples according to $\lF_{\bar f}$, where
\begin{equation*}
  \bar f :=
 \begin{cases}
   f(x), &\text{ if $x \notin P_0$} \\
   \eta, &\text{ if $x \in P_0$.}
 \end{cases}
\end{equation*}

For that, we have to run two ``copies'' of the algorithm in parallel, but the random decisions must be coupled; cf Appendix~\ref{apx:coupling} for the technical details.
We say that the two runs \textit{diverge,} if the algorithms' inner states differ, or one of the queries result in a different answer.

\begin{lemma}\label{lem:divergence}
  Under the conditions of Lemma~\ref{lem:tech}, the probability that the coupled runs of the algorithm diverge at or before the $m$th query is at most $\delta m/2$.
\end{lemma}
\begin{proof}
  As the random decisions of the two runs are coupled, for the runs of the algorithm to diverge, at least one sample-query has to give a different result for $\lF_f$ from $\lF_{\bar f}$.  If the first~$j$ query points and query results are the same, then, by coupling, the $(j+1)$th query point $x_{j+1}$ is the same, too.  Hence, the probability that the $(j+1)$th query result differs is $\abs{f(x) - \bar f(x)}/2$, which by by condition~\ref{lem:tech:cond:near-const} in Lemma~\ref{lem:tech}, is at most $\delta/2$.

  By induction, the probability that no query point or query result is different in any of the first~$m$ queries is at least
  \begin{equation*}
    (1-\delta/2)^m
  \end{equation*}
  which, by Bernoulli's inequality, is at least $1 - \delta m/2$.
\end{proof}

Now we can finish the proof of Lemma~\ref{lem:tech}.
\begin{proof}[Proof of Lemma~\ref{lem:tech}]
  W.l.o.g., we impose on algorithm \Algo{} that it queries (at least once) the point it eventually ouputs.   This might lead to an additional query, which we account for by the ``$+1$'' on the RHS of~\eqref{eq:ub-on-nqueries}.

  With this modification, the for~\Algo{} to output a point~$x^*$ with $f(x^*)\ge \max_x f(x) - \alpha$, by conditions \ref{lem:tech:cond:max} and \ref{lem:tech:cond:diff}, it is necessary that the algorithm queries at least~1 point $x\notin P_A$.  We will lower bound the random variable $T'_A$ which counts the number of queries of~\Algo{} until (and including) the first query point outside of $P_0$ is requested.

  As explained above, we now synchronize the run of \Algo{} with a coupled run where the samples are taking according to $\lF_{\bar f}$ instead of $\lF_f$.

  Let $m\in\NN$.  We distinguish tow cases.
  \begin{enumerate}[label=\arabic*.)]
  \item\label{sdlkfj:div} The two runs diverge before or at the $m$th query;
  \item\label{sdlkfj:not-div} The two runs do not diverge before or at the $m$th query.
  \end{enumerate}

  In case \ref{sdlkfj:div}, nothing can be said about $T'_A$.  The probability of this happening is at most $\delta m/2$ by Lemma~\ref{lem:divergence}.

  As for the case \ref{sdlkfj:not-div}, note that \Algo{} with samples according to $\lF_{\bar f}$ is a randomized strategy for Bob playing the plateau game: After each ``yes'' answer from Alice, he throws a coin with probability of heads $(1-\eta)/2$, and proceeds dependant on the outcome.  Hence, by Lemma~\ref{lem:stupid-game}, the probability of querying a point not in $P_0$ is at most $pm$.

  This concludes the proof of Lemma~\ref{lem:tech}.
\end{proof}

\subsection{A simple family of PQCs}
We will now give a family of PQCs $C^{(n)}$, $n\in\NN$, on $n$ qubits, which implement functions $f_0^{(n)}\colon \Tn\to\RR$ to which will apply Lemma~\ref{lem:tech}.  In this section, the ``PQC'' will stand synonymous with ``PQC with parameters given in the form~\eqref{eq:ham-unitary} and measurement is of an observable with $\pm1$-eigenvalues.''

We start with a simple observation for easy reference.
\begin{lemma}\label{lem:tensor-PQC}
  Let $g\colon\Tn[n]\to\RR$, $h\colon\Tn[m]\to\RR$, and suppose that the expectation value functions of PQCs $C_f$ and $C_g$ are $f$ and~$g$ resp.  Then $f\colon\Tn[n+m]\to\RR\colon (x,y)\to g(x)h(y)$ can be obtained as the expectation value function of a PQC.
\end{lemma}
\begin{proof}
  Let $C_g$ ($C_h$) use $q_g$ ($q_h$) qubits, perform the quantum operation $\Ee_g(x)$ on parameter setting $x$ ($\Ee_h(y)$ on parameter setting~$y$), and ultimately measure the observable $M_g$ ($M_h$).  With $\rho_0 := \ket{0}\bra{0}$:
  \begin{align*}
    g(x) &= \tr(M_g \; \Ee_g(x)  \; \rho_0) \\
    h(y) &= \tr(M_h \; \Ee_h(y)  \; \rho_0).
  \end{align*}
  The function~$f$ is realized as an expectation value function by taking $q_g+q_h$ qubits, staring in $\rho_0\otimes \rho_0$, applying the quantum operation $\Ee_g\otimes \Ee_h$, and utlimately measuring the observable $M_g\otimes M_h$.
\end{proof}

In view of Lemma~\ref{lem:tensor-PQC}, we first give a PQC for the function
\begin{equation}\label{eq:delta:1}
  h
  \colon \Tone \to \RR
  \colon x \mapsto \nfrac13 +\cos(2\pi x)/3.
\end{equation}
As for the first part, it can easily be seen that there exists a $\phi\in \lt]0,\pi/4\rt[$ such that
\begin{equation*}
  R := e^{-2\pi i(\cos\phi X + \sin\phi Z)}
\end{equation*}
has the property that, with $x:=\nfrac13$, the quantum operation $\rho \mapsto R^x\rho R^{-x}$ maps~$Z$ into the $X$-$Y$-plane.  The same is then the case for $x=\nfrac23$.
Now take the following PQC $C_1$:
\begin{equation}\label{eq:atomic-circuit}
  \Qcircuit{
    \lstick{\ket0} & \gate{R^{x}} & \meter \\
  }
\end{equation}
where the measurement at the end is of the observable~$Z$.
The expectation value function of~\eqref{eq:atomic-circuit} and~$h$ coincide for $x=0,\nfrac13,\nfrac23$, and since both are trigonometric polynomials of degree~1 (cf.~\cite{GilVidal-Theis:CalcPQC:2018}), they are the same for every $x\in\RR$.  Thus we have realized the function~$h$ in~\eqref{eq:delta:1} as the expectation value function of the parameterized quantum circuit~\eqref{eq:atomic-circuit}.

So we can realize the function~$h$ defined in~\eqref{eq:delta:1}.  Proving that~$f$ as defined in~\eqref{eq:delta:n} has the required properties requires a little work.

Second, we invoke Lemma~\ref{lem:tensor-PQC}, to give us a PQC, $C_n$, whose expectation value function is equal to
\begin{equation}\label{eq:delta:n}
  f^{(n)}_0
  \colon \Tn \to     \RR
  \colon x   \mapsto \prod_{j=1}^n h(x_j).
\end{equation}

\subsection{Conclusion of the proof of Therorem~\ref{thm:sample}}
We are now ready to conclude the proof of Theorem~\ref{thm:sample}.

Just one more definition: For $x\in\Tn$ and $a\in\dTn$, we denote by $d(a,x)$ the number of $j$ with $|x_j-a_j|<\nfrac16$, where for $y\in\RR$, $|y|$ denotes the \textit{Bohr ``norm'':} the smallest element in $y+\ZZ$.  The quantity $d(a,x)$ is the Hamming distance between $a$ and the point resulting from ``rounding''~$x$ to the closest element in $\dTn$ (if ties are broken properly).

With $f_0$ defined as in~\eqref{eq:delta:n}, we let\footnote{No attempts have been made to optimize the constants --- not even the base of the exponentially many queries.}
\begin{equation}\label{eq:def:P_0}
  P_0 := \lt\{ x\in\Tn \mid d(0,x) > n/2 \rt\}.
\end{equation}

We note the following simple fact-plus-definition (of $\delta$) for easy reference.
\begin{lemma}\label{lem:def-+-ieq--delta}
  For all $x\in P_0$,
  \begin{equation*}
    \abs{ f_0(x) } < (2/3)^{n/2}  =: \delta
  \end{equation*}
\end{lemma}
\begin{proof}
  This follows directly from the definitions of $f_0$ in~\eqref{eq:delta:n}, the fact that $\abs{h(t)} \le \nfrac23$ holds for all $t\in \Tone$ with $|t|\ge\nfrac16$, and the definition of $P_0$.
\end{proof}

The other important quantity for the use of Lemma~\ref{lem:tech} is~$p$, as defined in~\eqref{eq:def-p}.  Another short lemma-plus-definition.

\begin{lemma}\label{lem:def-+-ieq--p}
  For all $x\in\Tn$, if $A$ is chosen unformly at random in $\dTn$, then
  \begin{equation*}
    \Prb\bigl( P_A \not\ni x \bigr) \le e^{-n/36} =: p.
  \end{equation*}
\end{lemma}
\begin{proof}
  Fix an arbitrary $x\in\Tn$, and consider the Bernoulli random variables
  \begin{equation*}
    D_j :=
    \begin{cases}
      1, & \text{if $|x_j - A_j| < \nfrac16$},\\
      0, & \text{otherwise.}
    \end{cases}
  \end{equation*}
  Note that the $x\notin P_A$ is equivalent to
  \begin{equation*}
    D := \sum_{j=1}^n D_j > n/2.
  \end{equation*}

  Since the $A_j$ are independent, the $D_j$ are, too.
  Moreover, we have $\Prb( D_j = 1 ) \in \{0,\nfrac13\}$.  By an appropriate version of Hoeffding's inequality,
  %
  %
  we find that
  \begin{align*}
    \Prb\bigl( P_A \not\ni x \bigr)
    &= \Prb( D > n/2 )
    \\
    &\le \Prb( D > \Exp D + n/6 )
    \\
    &\le e^{-n/36}.
  \end{align*}
\end{proof}

We are now ready to put it all together.

\begin{proof}[Proof of Theorem~\ref{thm:sample}]
  With $f_0,P_0,\delta,p$ as above and $\eta := 0$, $\alpha := 1-2\delta$, the conditions of Lemma~\ref{lem:tech} are satisfied, with $\max_x f_0(x) = 1$.  Hence, for any $\alpha$-successful algorithm \Algo{} we have for the number of sample queries, $T_A$, satisfies~\ref{eq:ub-on-nqueries}.  By the definitions of $\delta,p$, we can find an absolute constant $c > 1$ such that
  \begin{equation}\label{eq:prb-bd-numq}
    \Prb( T_A < c^n ) \le c^{-n}.
  \end{equation}
  This concludes the proof of the theorem.
\end{proof}

\subsubsection*{Notes about Theorem~\ref{thm:sample}}
It should be pointed out that in our setup, the algorithm is given the magic ability to know whether it has queried a point with a close-to-maximal function value.  Hence, our definition of ``omnipotence'' may be modified to include that ability.

Theorem~\ref{thm:sample} is set in Las Vegas; it can readily be moved to Monte Carlo by just rearranging the proof.

The bound on the mutual information, \eqref{eq:sample-q:mutinf-exp-small}, can be obtained from~\eqref{eq:ub-on-nqueries} using standard tools; we leave that to the reader.

\section{Outlook}
It is argued in~\cite{Harrow-Napp:low-depth-gradient:2019}, that sampling from derivatives gives and advantage in terms of convergence of gradient descent algorithms, under some assumptions.  It remains to be seen whether querying samples from derivatives of~$f$ removes the exponentially poor information content of sample queries in setting.

\subsection*{Acknowledgements}
The authors would like to thank Vitaly Skachek for a discussion enabling the proof of Lemma~\ref{lem:stupid-game}, which simplified the whole \S\ref{ssec:sample-queries:game}.

This research was supported by the Estonian Research Council, ETAG (\textit{Eesti Teadusagentuur}), through PUT Exploratory Grant \#620.  Ironically, BG was partly supported by United States Air Force Office of Scientific Research (AFOSR) via AOARD Grant ``Verification of Quantum Cryptography'' (FA2386-17-1-4022).


\appendix
\section*{APPENDIX}
\section{Non-adaptivity of the Plateau Game}\label{apx:non-adapt}
Formally, the set of \textsl{adaptive} strategies of Bob would be a sequence $\mu_j$ of mapps from the set
\begin{equation*}
  \Bigl(  \Tn \times\{\texttt{Yes},\texttt{No}\} \Bigr)^{j-1}
\end{equation*}
of his  previous question and  Alice's answers to the set of probability distributions  on $\Tn$.  Since the right-hand entries (Alice's answers) are always the same (unless Bob has won), this collapses to: for every $j\in\NN$ and $x\in(\Tn)^{j-1}$, a probability distribution on $\Tn$ --- in other words, a probability distribution on $(\Tn)^j$.
\section{Coupling of the two runs of the algorithm}\label{apx:coupling}
We assume an infinite stack of independent random numbers $R_0,R_1,R_2,R_3,\dots$, each uniformly distributed in $[-1,+1]$.   Based on them, the sampling from $\lF$ is performed and the random decisions of the algorithm are made:
\begin{itemize}
\item When querying a sample from $\lF_g(x)$, the top of the stack, $R_0$ is considered.  If $R_0 < g(x)$, the query result is~$+1$, otherwise it is $-1$.  $R_0$ is then popped from the stack.
\item Similarly, when the randomized algorithm tosses a coin, the top $R_0$ is consulted in the obvious manner, and popped.
\end{itemize}

We run the algorithm \Algo{} simultanously with two different sample-query distributions: In one of the answer the queries according to the distributions $\lF_f$, in the other run we answer according to $\lF_{\bar f}$.  In the two runs, we use two stacks the same infinite sequence of random numbers.

\end{document}